\documentclass{article}
\usepackage{authblk}
\usepackage[utf8]{inputenc}
\usepackage[centertags]{amsmath}
\usepackage{newlfont,color}
\usepackage{amsmath,amsthm,amsopn,amstext,amscd,amsfonts,amssymb}
\usepackage{cases}
\usepackage{multicol} 
\usepackage{longtable}
\usepackage{float}
\usepackage{graphicx}
\usepackage{url}
\usepackage{arydshln}
\usepackage{subfig}
\usepackage{multirow}
\usepackage{commath}
\usepackage{enumitem}
\usepackage{cancel}
\usepackage{anysize}
\usepackage[spanish, USenglish, activeacute]{babel}
\usepackage{mathrsfs}
\usepackage{enumitem}
\usepackage[skip=5pt]{caption}
\usepackage[title,titletoc,toc]{appendix}
\newtheorem{corollary}{Corollary}[section]
\newtheorem{lemma}{Lemma}
\newtheorem*{lemma*}{Lemma}

\newtheorem{theorem}{Theorem}
\newtheorem*{theorem*}{Theorem}

\newtheorem{observation}{Observation}
\newtheorem*{observation*}{Observation}
\definecolor{charcoal}{rgb}{0.21, 0.27, 0.31}
\marginsize{2 cm}{2 cm }{2 cm}{2 cm}
\title{Partisan affect and political outsiders}
\author[]{Fernanda Herrera\thanks{I gratefully acknowledge financial support from the University of California Institute for Mexico and the United States UC MEXUS, and Mexico's National Council of Science and Technology CONACYT, through FONCICYT. I also wish to express my gratitude to David Cantala for the helpful suggestions during the preparation of the paper.}}
\affil[]{School of Global Policy and Strategy, University of California, San Diego}
\affil[ ]{\normalsize{fherrera@ucsd.edu}}
\date{}

\providecommand{\keywords}[1]{\textbf{{Keywords}} #1}

\begin{document}
	
\maketitle

\begin{abstract}
We examine the effects of introducing a political outsider to the nomination process leading to an election. To this end, we develop a sequential game where politicians -insiders and outsiders- make a platform offer to a party, and parties in turn decide which offer to accept; this process conforms the voting ballot. Embedded in the evaluation of a party-candidate match is partisan affect, a variable comprising the attitudes of voters towards the party. Partisan affect may bias the electorate's appraisal of a match in a positive or negative way. We characterize the conditions that lead to the nomination of an outsider and determine whether her introduction as a potential candidate has any effect on the winning policy and on the welfare of voters. We find that the victory of an outsider generally leads to policy polarization, and that partisan affect has a more significant effect on welfare than ideology extremism. 
\end{abstract}

\\


\section{Introduction}
While individuals with no political experience, henceforth outsiders or newcomers, are not new to the political arena, their surge and increasingly frequent success in congressional and presidential elections poses interesting questions from the perspective of voters, politicians and parties. Why do voters elect individuals with no proven governing skills, particularly for high office? What makes citizens want to leave the market sector and enter politics? And what makes parties endorse outsiders over members from their own pool? 

The value of political experience in electoral performance has led to a vast literature on the benefits of serving in office and tenure, specifically in the form of incumbency advantage. See Erikson (1972) and Born (1979) for an early account on the subject, and Alford and Hibbing (1981) for a refinement of electoral career patterns. The discussion on incumbency advantage has mainly revolved around what constitutes the advantage itself, whether it is attributable to candidate quality or that of the challengers, officeholder benefits or constituency service, and how to best measure it (Levitt and Wolfram, 1997; Cox and Katz, 1996; King, 1991; Gelman and King, 1990). Although being an incumbent can be a determining factor for achieving electoral success, a less explored but equally fundamental question is what accounts for the success of candidates with no prior political experience? 

The issue has gained relevance in recent times, since the emergence of a pattern of victories of newcomers -in many cases over well known politicians with consolidated careers- from what used to be a handful of noticeable anecdotes. Porter and Treul (2020) find that from 2010 to 2018, inexperienced candidates have been more successful than ever before, particularly in the period from 1980 to 2014, in winning primary elections for the U.S. House of Representatives, both in incumbent and non-incumbent seat races.
Further demonstrating the threat that newcomers constitute, Roberds and Roberts (2002) find that in House elections between 1992 and 1998, incumbent politicians spent nearly the same amount campaigning against quality challengers as they did running with competitive amateurs.\footnote{Their definition of quality challenger is a politician with previous elective office experience. They follow Canon (1993) in defining a competitive amateur as a celebrity, a repeat challenger, or a candidate that defeated at least one quality candidate in the primary election.} The consolidation of newcomers as serious and competitive challengers has derived in an incipient strand of literature addressing the seeming change in the preference of voters for (in)experience, the qualities and motivations of outsiders, their previous professional experience, and aims and intentions for entering the political sector (Hetherington and Rudolph, 2015; Porter and Treul, 2020; King, 2002; Mattozzi and Merlo, 2008).

Among the factors explaining the election of inexperienced leaders, we highlight the electorate's mistrust or discontent towards career politicians or the government, which serve only to increase the appeal of newcomers. Carreras (2017) shows that newcomers have electoral success when macroeconomic conditions are bad and when there is a high level of corruption. Popular discontent, derived from economic or political crisis, not only increases the likelihood of victory of outsiders, it also escalates their extremism (Serra, 2018). While this sheds light on the question of why voters may elect outsiders, it leaves the issue of why parties choose to endorse one over its own insiders unresolved. Even if one ventures to guess it is in the spirit of winning, a natural question arises: knowing that her participation is indispensable to this end, what conditions does the outsider impose to associate herself with a party? With this work we aim to elucidate this matter and to explore the consequences of having an outsider as a political contender on the welfare of voters and the party.

We develop a model of nomination and election where politicians, both insiders and outsiders, propose a platform to a party. A platform is a binding pair of a policy and the rent the politician will seek in office if she is elected. After evaluating the offers they receive, the parties decide which platform to accept, and consequently what match to make for the election. Note that the nomination process is not strictly a primary election, although it can be interpreted as such. A key point is that politicians make a platform offer only once and at the beginning, meaning that they think about appealing to the party elite and voters simultaneously. Our intuition is that a party endorses an outsider when it is in need of a `fresh start', that is, when for a given reason voters associate it and its members with undesirable traits or deeds. The sentiment, however is not transferred to outsiders for a number of reasons. One is that if an association between a party and the outsider is made, it is too recent to imprint (dis)affect, another is that outsiders seem to be distinctively good at setting themselves apart from the power elite and parties, often embodying anti-establishment or anti-incumbent rhetoric (Bar, 2009). 

We incorporate the attitudes of voters towards parties through a concept called partisan affect. It is a generalization of Fowler and Hall's notion of partisan incumbency advantage,  defined as `the electoral benefit a candidate receives purely because her party is the incumbent party' (2014: 502). They also allow for negative incumbency advantage, which works in detriment of the candidate, and may be attributed to a preference for partisan balance over time, or an overestimation of outside options, namely opposition candidates. Our definition is too a transfer of benefits or liabilities from a party to its members, although we let this be true for both parties, not just the incumbent. Thus, partisan affect has the ability to enhance (or diminish) the voters' regard of insiders. This is similar to Serra's (2018) model, where popular discontent impacts candidate charisma, although in his model, the outsider is effectively the only politician affected by it.\footnote{This is because in his model, charisma stands for \textit{populist charisma}, an anti-establishment quality that is positively valued by voters, and possessed credibly only by the outsider.} By allowing the electorate to like or dislike one party or both, we are able to model a range of scenarios where a match with an outsider is undesirable, an alternative to victory, or a party's last resource option.

Our work is similar to van Weelden's (2013) re-election incentives model, which explains why voters (re)elect non-median candidates. He follows a citizen-candidate tradition, in the sense that voters as well as candidates care about the implemented policy at all times, that is, during and after the term in office. This is crucial for disciplining incumbents through the threat of electing politicians on the other side of the political spectrum if they engage in undesirable behavior. In our model, the assumption of policy awareness means that outsiders have incentives for seeking nomination to avoid the election of candidates that will implement policies far from their ideal point; whether that is insiders from the same ideological side or from the opposition. Van Weelden also deals with rent-seeking in office, which takes the form of corruption or lack of effort, and is consequently beneficial to the elected candidate and detrimental to the electorate. We let only the insider seek rent, as a way of making up for her outside options, although alternatively it can be interpreted as an ego rent for succeeding in a brand new enterprise. Departing from van Weelden, we make the nominating party accountable for any rent-seeking, thus the transfer comes from it and not voters. The rationale is that when a party nominates an outsider, it sends a message to voters, namely that its insiders are not good enough, so the match carries a credibility cost. The worst the outsider performs, the worst the party comes off, as it appears it is the best it had to offer.

We find that compared to a model where there are only insider politicians, the victory of an outsider generally leads to a more extreme winning policy, although there are a few conditions that lead to a different outcome. The victory of the party's insider has mixed effects on the the polarization of the policy. The key aspect of our results is that a party that receives two platform offers -one from its insider and one from the outsider- has always the option to accept either, and this potentially limits policy swaying or excessive rent-seeking on the advantaged candidate's behalf. Of course, the less advantaged politician may be in such a precarious condition that her offer has virtually no effect on her opponent's successful platform range. We also find that the introduction of the outsider may improve or diminish voter welfare, while in most cases it preserves or improves that of the nominating party. This is true even if the outsider engages in full rent, so at best, a victory with an outsider is better than losing, even if there is a credibility cost to pay. All our results are independent of the outsider's ideology, which suggests that candidate valence in conjunction with partisan affect (relative to that of the contenders) and not extremism, is the most important driver of welfare.

The paper proceeds as follows. In Section \ref{model} the model is presented, and its assumptions briefly discussed. Section \ref{solution} gives the solution to two main election scenarios: one in which insiders are the only contenders, and another in which an outsider enters the game. The presentation of the solution is in the order dictated by backward induction, and intuition is provided when pertinent. In Section \ref{discussion} we discuss the findings of the model, namely the effects of outsiders on the winning policies and on the welfare of voters and parties. We also compare our results with those of akin work. Concluding remarks and possible extensions of the model are in Section \ref{conclusion}. 
\section{Model}\label{model}

\subsection{Basic setting}

Ideologies are represented by a point in the real line. The bliss point of the median voter, who is the decisive voter in the election, is located at 0. Leftist standpoints lie in $(-\infty,0)$, and rightist ones in $(0,\infty)$. There are two parties $L$ and $R$, both of which are characterized by an ideological bliss point, denoted respectively $b_L$ and $b_R$, with $b_L<0<b_R$. A party is generically denoted by $P$, and its bliss point by $b$.

Both parties elicit in voters some degree of affection or disaffection, and this is comprised in the variables $\alpha_L,\alpha_R\geq0$, which we refer to as partisan affect factors. Affection may arise from regular attendance of party representatives to House votes, the implementation of transparency policies, or the effectiveness of disciplinary committees to impose sanctions on party members who behave unethically. Disaffection may be due to scandals related to corruption, the preference of voters for rotation of parties in power, or even random shocks that have a negative effect on the economic conditions in a constituency.\footnote{Wolfers (2002) finds that voters in oil producing states tend to re-elect incumbent governors during price rises and vote them out in price drops.} It is worth pointing out that in our model, the affect factors are uncorrelated with the performance of current or past politicians in office. This is meant to reflect that the affect factors are in fact strictly partisan, meaning there is no personal (dis)advantage transference from politicians to the parties. The assumption is also made for convenience, as it allows us to treat the past independently from the present, so we need not keep track of the history of previous winners. This is why $\alpha_L$ and $\alpha_R$ are taken to be exogenous. 

Each party has an affiliated politician, called insider, denoted by $l$ and $r$ respectively. There is also an outsider politician $o$, who is a political neophyte with no party affiliation that is looking to be nominated by a party to run for office, say $R$. We consider lack of experience to be the main characteristic of an outsider, not her momentary lack of affiliation. This is to say, independent politicians are typically not outsiders, since in general, they have experience either as former party members, or in elections. All politicians are characterized by a bliss point along the ideology dimension, with $k_l<0<k_r$, and $k_o\geq0$,\footnote{Note that we allow for a centrist outsider.} and a set of non-ideological attributes valued positively by all voters, comprised in a variable called valence, with $\nu_l,\nu_r,\nu_o\geq0$. We regard valence as a feature of politicians that is useful in elections, it can be interpreted henceforth as charm or charisma.\footnote{Other scholars also comprise in valence some features that are useful for governing (see Adams \textit{et al}, 2011). We will not account for this dimension, but refer the reader to Serra (2018) for a discussion of the effect of popular disaffection on the decision of a charismatic politician of how much political experience to gain.}

Politician $p$ offers her party a platform $(x_p,m_p)$, which consists of a policy $x_p$, and the rent $m_p\geq0$ she will seek in the event she is elected for office. It is up to the party to decide whether to take the offer, and thus endorse the candidate, or not. Each party must endorse strictly one candidate for the election. Policies are binding, meaning that they are straightforwardly implemented by the winner of the election. Policy offers are such that $x_l\leq0$, and $x_r,x_o\geq0$.

We follow Van Weelden in assuming that rent is lack of effort on the politician's behalf ``either taking time off or not addressing difficult but important non-ideological issues'' (Van Weelden, 2013).\footnote{An example of rent-seeking is the frequent golfing Donald Trump has engaged in during his presidency. Solender (2020, July 12) estimates that Trump has gone golfing more than 1 in every 5 days in the course of his term. His arrival at the Oval Office too has been reported to be around noon, and it has been attributed to watching television (Rogers, K. and Annie Karni, 2020, April 23).} We make the assumption that the outsider is more prone to seeking a positive rent than insiders, either because she is more likely to ignore political norms and processes -she is less aware or concerned by peer pressure-, or has less vocation. Indeed, a higher sense of calling would translate (for the most part) into an earlier incursion into politics. To highlight this, we set $m_l=m_r=0$, and let $m_o\geq0$. Lack of effort is costly to the party because the slacking of a politician is picked up by the administration or her political staff, and this carries a cost in terms of time and resources.\footnote{Rent-seeking could also pose a reputation or credibility-cost to the incumbent party. Since we have assumed that the performance of politicians does not affect the appraisal of a party, we do not develop this possibility further.} We set the cost of having a rent-seeking politician in office to be $c=m_p$, so the party fully restores effort to constituents.

Since $L$ receives only one platform offer, namely $(x_l,m_l)$, it accepts it, making politician $l$ its candidate $c_L$, and her proposed policy $x_l$, the campaign policy $x_L$. When deciding which platform to accept, $R$ takes into account the preference of voters, which may be influenced by the partisan affect factors $\alpha_L$ and $\alpha_R$. To be precise, we assume that the affect  factors influence the voters' appraisal of party-insider matches, that is, those between $L$ and $l$ and $R$ and $r$, but not that of a party with the outsider, $R$ and $o$. The reason being that insiders are more associated with a party than are outsiders, and this leads to a higher chance of affect transfer.\footnote{It is possible that insiders have previously run as candidates, either for the same or a lower office, and ``since candidates tend to campaign on and talk frequently about the issues their party owns, the public comes to associate certain traits with that party’s candidates'' (Hayes, 2005: 908).} Upon accepting the offer $(x_p,m_p)$ with $p$ in $\{r,o\}$, politician $p$ becomes $R$'s candidate $c_R$, and her policy $x_p$, the campaign policy $x_R$. We write $\left(P,c_P,(x_P,m_P)\right)$ for a party-candidate-platform match. 

One of our interests is to determine the conditions under which a party opts to run with an outsider as opposed to its insider politician, even when this is potentially costly to the party. We also investigate the effect of having an outsider contender on the welfare of voters and the nominating party. In our model, the utility functions, payoffs, and strategies of all agents are common knowledge. Hence, both the insider and the outsider know they compete for nomination through their platforms $(x_r,m_r)$ and $(x_o,m_o)$. Furthermore, politicians and parties know the preference of voters, hence they are able to determine which platforms appeal to voters best.

\subsection{Preferences}

\subsubsection{Voter preferences}
On election day, the electorate sincerely votes for the party-candidate-platform match they prefer in the ballot $\{\left(L,c_L, (x_L,m_L)\right), \left(R,c_R,(x_R,m_R)\right)\}$, with $c_L=l$, $x_L=x_l$, $m_L=0$, and $c_R$ in $\{r,o\}$, $x_R$ in $\{x_r,x_o\}$ and $m_R$ in $\{0,m_o\}$. A party-candidate-platform match $\left(P,c_P,(x_P,m_P)\right)$ gives utility to voters through two dimensions: the campaign policy $x_P$, and the valence associated to candidate $c_P$, which is $\nu_{c_P}$. 

We deem straightforward to assume that the electorate values a given policy equally, regardless of whether it is proposed by an insider politician or an outsider.\footnote{It is possible that the voters' trust on a candidate's ability to implement a policy depends on her political experience. Recall that we have not accounted for the governing aspect of valence in this model, so this case is beyond our current analysis. However, a reassuring observation of Kitschelt (2000) is that in presidential systems, candidate competition is through personal charisma and not policy programs.} Similar to Serra (2018), we assume that the effect of the valence of candidates is mediated by the partisan affect factor, so voters perceive the valence of $P$'s candidate $c_P$, namely $\nu_{c_P}$, to be $\alpha_P\nu_{c_P}$.\footnote{Serra works with charisma $c$ instead of valence $\nu$, and with popular discontent $\delta$ instead of partisan affect $\alpha$. The perception of the valence of a candidate in Serra's model is $\delta c$.} 

Generally speaking, the utility to the median voter $M$, over the match $\left(P,c_P,(x_P,m_P)\right)$, is defined to be 
\begin{equation}\label{utM}
U_M\left(P,c_P,(x_P,m_P)\right)=-\abs{x_P}+\alpha_P\nu_{c_P},\quad\text{with $\alpha_P\equiv1$ if $c_P$ is the outsider.}
\end{equation}

We thus have,
\begin{eqnarray*}
U_M\left(L,l,(x_l,0)\right)&=&-\abs{x_l}+\alpha_L\nu_l,\\
U_M\left(R,r,(x_r,0)\right)&=&-\abs{x_r}+\alpha_R\nu_r,\\
U_M\left(R,o,(x_o,m_o)\right)&=&-\abs{x_o}+\nu_o.
\end{eqnarray*}

\noindent\textit{Indifference assumption}: If the matches $\left(L,l,(x_l,0)\right)$ and $\left(R,c_R,(x_R,m_R)\right)$ with $c_R$ in $\{r,o\}$ yield the same utility to $M$, then she randomizes her vote in such a way that either match wins the election with equal probability.\\

\noindent\textbf{Political agent preferences}
\vspace{0.25cm}

\noindent We follow Osborne and Slivinski (1996) in assuming that just like ordinary citizens, all the politicians care about the winning policy. Since party elites are politicians themselves, we expect parties $L$ and $R$ to be policy-aware as well.

\subsubsection{Politician preferences}

The utility that a politician derives from proposing a given platform is contingent on whether it is accepted by the party or not, and on the subsequent result of the election. The election determines if she will get the rent she seeks, and the winning policy upon which she draws utility.

Assume that politician $p$ with bliss point $k_p$, proposes the platform $(x_p,m_p)$ to $P$. There are two main cases to consider.

\begin{enumerate}
	\item The offer of $p$ is accepted by $P$.
	
	\begin{enumerate}[label=(\alph*)]
		\item If $p$ ends up winning the election,
		\begin{equation*}
		U_p(x_p,m_p)=-\abs{x_p-k_p}+m_p.
		\end{equation*}
		\item If $p$ loses the election,
		\begin{equation*}
		U_p(x_p,m_p)=-\abs{y-k_p},\hspace{0.5em}\text{where $y$ is the winning policy.}
		\end{equation*}
	\end{enumerate}

	\item The offer of $p$ is rejected by $P$.
	
	Note that $R$ is the only party that can reject offers, thus $p$ in $\{r,o\}$. We have:
	\begin{equation*}
		U_p(x_p,m_p)=-\abs{y-k_p},\hspace{0.5em}\text{where $y$ is the winning policy.}
	\end{equation*}
It is clear that having an offer rejected yields the same utility as being endorsed by a party and losing the election; this is because no monetary costs to running were included in the model.\footnote{Dal Bo \textit{et al.} 2017 point out that the costs of running are usually covered by parties or outside donors rather than politicians themselves. So, at least on this end, not including them does not alter the preferences of politicians.}
\end{enumerate}

\noindent\textit{Indifference assumption}: If two platforms $(x,m)$ and $(x',m')$ yield the same utility to a politician, she randomizes her offer in such a way that either platform is proposed with equal probability.

\subsubsection{Party preferences}

We define the utility that party $P$, with bliss point $b$, derives from the offer $(x_p,m_p)$ of politician $p$ to be:

\begin{align}\label{partypref}
U_P(x_p,m_p)=&
\begin{cases}
-\abs{x_p-b}-m_p& \text{if $p$ is endorsed and wins the election,}\\
-\abs{y-b}&\text{if $p$ is endorsed and loses the election, and $y$ is the winning policy.}
\end{cases}
\end{align}
It should be stressed that the winning policy $y$ is potentially affected by the politician that $P$ ends up running with, we will elaborate this point further in the next section.\footnote{It suffices to say that if party $R$ anticipates losing the election, then it strategically determines which candidate to endorse, as to minimize the distance between $y$ and its own bliss point $b_R$.}

\noindent\textit{Indifference assumption}: If two offers $(x,m)$ and $(x',m')$ yield the same utility to a party, then it accepts that of the insider. 

Apart from being convenient, this assumption is meant to reflect that parties know insiders better, and for this reason, trust them more. In any case, the assumption is not restricting and can be easily relaxed without significantly altering the model.

\subsection{Timing}

The utilities, valences, and the affect factors are all common knowledge.

\begin{enumerate}
	\item Politicians make platform offers $(x,m)$ to their parties. That is, $l$ offers $(x_l,0)$ to $L$, and $r$ and $o$ offer $(x_r,0)$ and $(x_o,m_o)$ to $R$ respectively. All of the offers are made simultaneously.
	
	\item Parties simultaneously select their preferred platform offer, and endorse the candidate associated to it. At this point, the  party-candidate-platform matches $\left(L,l,(x_l,0)\right)$ and $\left(R,c_R,(x_R,m_R)\right)$ with $c_R$ in $\{r,o\}$, $x_R$ in $\{x_r,x_o\}$ and $m_R$ in $(0,m_o)$ are made.
	
	\item The electorate sincerely votes for their preferred match.
\end{enumerate}

It is worth noting that only $l,r,o$, and $R$ make strategic choices. Indeed, the electorate votes sincerely, so it just picks the match that maximizes its utility. Party $L$ receives only one platform offer, and by convention, accepts it. In what follows, we derive the solution of the game, which is subgame-perfect equilibrium, by backward induction. The main results are in Subsection \ref{eqplat}.

\section{Solution}\label{solution}

We first find from the voters' decision making, the range of policies that parties deem acceptable. We then find the optimal strategies of the politicians $l$, $r$ and $o$ within the specified range. That is, we determine the equilibrium offers $(x^*_l,0)$, $(x^*_r,0)$, and $(x^*_o,m^*_o)$. With this, we work out which offer the parties end up accepting, and the match that ultimately wins the election.

Let us denote by $X=\abs{x_R}-\abs{x_L}$,\footnote{Or equivalently, $X=x_R+x_L$, since $x_L\leq0\leq x_R$.} the relative policy divergence from the origin, and by $V_{c_R}=\alpha_R\nu_{c_R}-\alpha_L\nu_L$ with $\alpha_L,\alpha_R\geq0$, the relative valence advantage of $R$'s candidate $c_R$, with respect to that of $L$, $c_L$.
With this notation, $V_r=\alpha_R\nu_r-\alpha_L\nu_l$ stands for $R$'s relative valence advantage of endorsing $r$ (with respect to $l$), and $V_o=\nu_o-\alpha_L\nu_l$ is that of endorsing $o$. We keep our analysis general, and deal with $V_{c_R}$ for any candidate $c_R$, unless otherwise stated.

We emphasize that by definition, and despite constraints on $x_L$ and $x_R$, $X$ and $V_{c_R}$ can take value in all the real line. Therefore, we study the entire coordinate plane $V_{c_R} X$.

\subsubsection{Outcome of the election}

\begin{observation}\ \\
If $X=V_{c_R}$ with $c_R$ in $\{r,o\}$, then the median voter is indifferent between the match $\left(L,c_L,(x_L,m_L)\right)$, and $\left(R,c_R,(x_R,m_R)\right)$. If $X>V_{c_R}$, the median voter prefers $\left(L,c_L,(x_L,m_L)\right)$ over $\left(R,c_R,(x_R,m_R)\right)$, and if $X<V_{c_R}$, she prefers $\left(R,c_R,(x_R,m_R)\right)$ over $\left(L,c_L,(x_L,m_L)\right)$.
\end{observation}

\begin{proof}\ \\
Note that $X=V_{c_R}$ implies $\abs{x_R}-\abs{x_L}=\alpha_R\nu_R-\alpha_L\nu_L$. Rearranging terms gives $$-\abs{x_L}+\alpha_L\nu_L=U_M\left(L,c_L,(x_L,m_L)\right)=-\abs{x_R}+\alpha_R\nu_R=U_M\left(R,c_R,(x_R,m_R)\right),$$ which is the conclusion. Similar arguments apply to the other two cases. This is shown in Figure \ref{fig:area}.
\end{proof}

\begin{figure}[h!]
	\centering
	\caption{Partisan appeal-induced preference regions}
	\includegraphics[width=7.5cm]{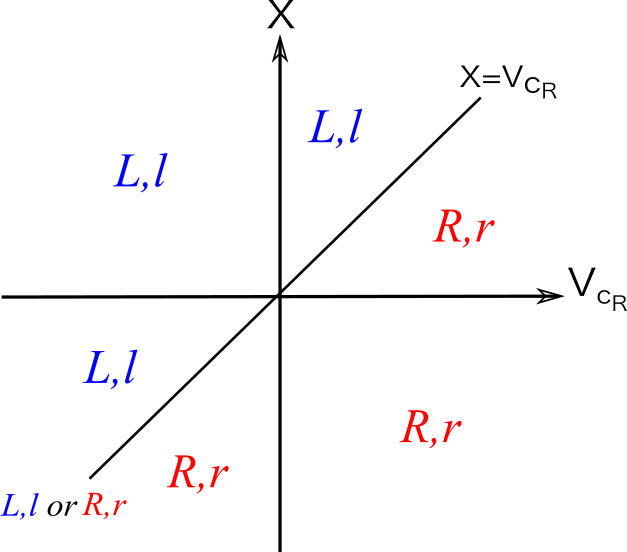}
	\label{fig:area}
\end{figure}
 
\subsubsection{Winning policy ranges}

\begin{lemma}\label{Lemma1}\ \\
	In the presence of relative valence advantage $V_{c_R}\geq0$, party $R$ accepts policies $x_R$ in $[0,V_{c_R}]$, and $L$ accepts $x_L\leq0$. The match $\left(L,l,(x_L,0)\right)$ loses the election to $\left(R,c_R,(x_R,m_R)\right)$, with $m_R\geq0$. In the presence of relative valence disadvantage $V_{c_R}<0$, $L$ accepts policies $x_L$ in $[V_{c_R},0]$, and $R$ accepts any $x_R\geq0$. In this case, the match $(R,c_R,(x_R,m_R))$ with $m_R\geq0$, loses the election to $(L,l,(x_L,0))$.
\end{lemma}
\begin{proof}\ \\
We refer the reader to Appendix \ref{AppA}.
\end{proof}

\begin{corollary}\label{cor1}\ \\
	If the valence advantage is null $V_{c_R}=0$, then the range of winning policies of $R$ converges to a single point, which is the bliss point of the median voter.
\end{corollary}

\begin{proof}
	The proof is obtained immediately substituting $V_{c_R}=0$ in Lemma \ref{cor1}.
\end{proof}
The interest of the corollary is in its interpretation: if neither $c_L$ nor $c_R$ enjoys a valence advantage (note that $V_{c_R}=0$ implies $\alpha_R\nu_R=\alpha_L\nu_L$), then the matches compete for votes purely through their policies $x_L$ and $x_R$. This is nothing but Hotelling-Downs (1929, 1957) model of two party competition, where both policies converge to 0.

\subsubsection{Equilibrium platforms}\label{eqplat}

We have established that the range of acceptable policies depends on the relative valence advantage factor $V_{c_R}$ with $c_R$ in $\{r,o\}$. It is turn to determine the optimal policies for politicians within each range, and their respective rent offers. Before we start comparing all the possible combinations of values of $V_r$ and $V_o$, we turn our attention to a more specific, yet equally important case: that which consists of only insiders $l$ and $r$. As it may be inferred, this simplified case involves only one range-conditioning variable, namely $V_r=\alpha_R\nu_r-\alpha_L\nu_l$. As we have mentioned before, one of our interests is to investigate how the entrance of an outsider affects the welfare of all the agents involved, so this is the base of our analysis. Equally important is the fact that the insider-only case lays the basic intuition for the solution of the more involved model with an outsider.

\begin{theorem}\label{Thm1}\ \\
	A game consisting of only insiders $l$ and $r$, has the equilibrium offers and winning match stated in Table \ref{tab:eq1}.
\end{theorem}

	\begin{table}[ht]
	\caption{Equilibrium platforms and winning match, insider game}
	\begin{center}
		\begin{tabular}{ | l| l | l| l |}
			\hline
			\textbf{\textit{Case}}& \textbf{\textit{Valence condition}} &\textbf{\textit{Platform offers}} & \textbf{\textit{Winning match}}\\ \hline
			
			\multirow{4}{*}{\quad\textit{1}} &\multirow{4}{*}{\qquad\quad$\boldsymbol{V_r< 0}$} & $x^*_l=\max\{k_l,V_r\}$& \textit{Party}: $P=L$\\
			& & $m^*_l=0$& \textit{Politician}: $c_L=l$\\
			& &$x^*_r=0$ & \textit{Policy}: $x_L=\max\{k_l,V_r\}$\\
			& & $m^*_r=0$& \textit{Rent}: $m_L=0$\\
			\hline 
			
			\multirow{4}{*}{\quad\textit{2}}& \multirow{4}{*}{$\qquad \quad\boldsymbol{V_r= 0}$}& $x^*_l=0$& \textit{Party}: $P=L/R$\\
			& & $m^*_l=0$& \textit{Politician}: $c_P=l/r$\\
			& & $x^*_r=0$& \textit{Policy}: $x_P=0$\\ 
			& & $m^*_r=0$ & \textit{Rent}: $m_P=0$\\
			\hline
			
			\multirow{4}{*}{\quad\textit{3}}& \multirow{4}{*}{$\qquad \quad\boldsymbol{0< V_r}$}& $x^*_l=0$& \textit{Party}: $P=R$\\
			& & $m^*_l=0$& \textit{Politician}: $c_R=r$\\
			& & $x^*_r=\min\{k_r,V_r\}$& \textit{Policy}: $x_R=\min\{k_r,V_r\}$\\
			& & $m^*_r=0$ & \textit{Rent}: $m_R=0$\\
			\hline			
		\end{tabular}
		\end{center}
		\label{tab:eq1}
	\end{table}

\begin{proof}\ \\
	We give the main idea for Case 1, $V_r<0$, similar arguments apply to the other two cases. By Lemma \ref{Lemma1}, if $V_r<0$, then the match $(R,r,(x_r,0))$ with $x_r\geq0$ is bound to lose the election to $(L,l,(x_l,0))$ with $x_l$ in $[V_r,0]$. Using the notation introduced in Section \ref{partypref}, this means that for $R$, $U_R(x_r,0)=-\abs{x_l-b_R}$ for all $x_r\geq0$. We adopt the convention that in cases like this, the losing candidate proposes a centrist policy, so $x^*_r=0$. By definition of rents, we have $m^*_l=m^*_r=0$. What is left is to compute the actual value of $x^*_l$.
	
	Note that if the bliss point of $l$ is in the range of acceptable offers for $L$, namely $k_l$ in $[V_r,0]$, then she offers it as her campaign policy $x_l=k_l$, since this maximizes her utility. This is show in Figure \ref{fig:eq1}. If $k_l<V_r$, then she chooses the policy closest to her bliss point on the interval $[V_r,0]$, which is $x_l=V_r$. This is shown in \figref{fig:eq2}. The two arguments together imply that if $V_r<0$, then $x^*_l=\max\{k_l,V_r\}$. 
	
	\begin{figure}[!tbp]
		\centering
		\begin{minipage}[b]{0.4\textwidth}
			\includegraphics[width=\textwidth]{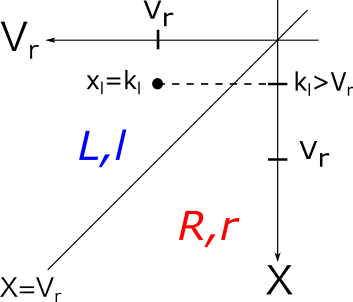}
			\caption{Equilibrium policy of $l$ when\\$V_r<k_l<0$}
			\label{fig:eq1}
		\end{minipage}
		\qquad
		\begin{minipage}[b]{0.4\textwidth}
			\includegraphics[width=\textwidth]{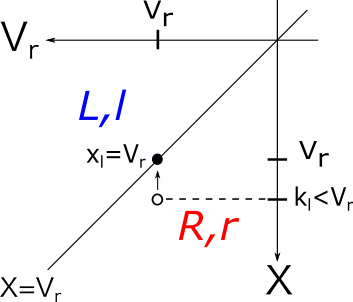}
			\caption{Equilibrium policy of $l$ when\\$k_l<V_r<0$}
			\label{fig:eq2}
		\end{minipage}
	\end{figure}
	
	Finally, we show that by accepting $x^*_l=\max\{k_l,V_r\}$, $L$ wins the election. Since $x^*_l=\max\{k_l,V_r\}$, we have $V_r\leq x^*_l$. Substituting the definition of $V_r$ into the last inequality gives $\alpha_R\nu_r-\alpha_L\nu_l<\max\{k_l,V_r\}$, rearranging terms and using the definition of absolute value we get: $U_M(R,r,(x^*_r=0,0))=\alpha_R\nu_r<-\abs{\max\{k_l,V_r\}}-\alpha_L\nu_l=U_M(L,l,(x^*_l=\max\{k_l,V_r\},0))$. This completes the proof.

 \end{proof}

Deriving the equilibrium offers and winners of a game with an outsider implies incorporating into the analysis the variables $V_o$, $k_o$, and $b_R$. As it may be inferred, the number of cases to consider grows exponentially fast. Fortunately, we are able to summarize them in a table similar to that of Theorem 1.

\begin{theorem}\label{Thm2}\ \\
A game consisting of insiders $l$ and $r$, and outsider $o$, has the equilibrium offers and winning match stated in Table \ref{tab:eq2}.

	\begin{center}
	\begin{longtable}{| l | l | l | l |}
		\caption[Equilibrium platforms and winning match]{Equilibrium platforms and winning match, game with outsider\label{tab:eq2}}\\
		\hline
		\textbf{\textit{Case}}& \textbf{\quad\textit{Valence conditions}} &\textbf{\textit{Platform offers}} & \textbf{\textit{Winning match}}\\
		\hline
		\endfirsthead
		\caption[]{Equilibrium platforms and winning match, game with outsider (continued)}\\
		\hline
		\textbf{\textit{Case}}& \textbf{\quad\textit{Valence conditions}} &\textbf{\textit{Platform offers}} & \textbf{\textit{Winning match}}\\
		\hline
		\endhead
		\hline
		\multicolumn{4}{|r|}{\small Continued on next page}\\
		\hline
		\endfoot
		\multicolumn{4}{|r|}{\small End of table}\\
		\hline
		\endlastfoot
		
		\multirow{4}{*}{\quad\textit{1}}& \multirow{4}{*}{$\quad\enspace\qquad\boldsymbol{V_r\leq V_o<0}$}& $x^*_l=\max\{k_l,V_o\}$& \textit{Party}: $P=L$\\
		& & $x^*_r=0$& \textit{Politician}: $c_L=l$\\
		& & $x^*_o=0$& \textit{Policy}: $x_L=\max\{k_l,V_o\}$\\
		& & $m^*_o=0$& \textit{Rent}: $m_L=0$\\
		
	   \pagebreak
		
		\multirow{4}{*}{\quad\textit{2} }& \multirow{4}{*}{$\quad\enspace\qquad\boldsymbol{V_o\leq V_r<0}$} &$x^*_l=\max\{k_l,V_r\}$& \textit{Party}: $P=L$\\
		& &$x^*_r=0$& \textit{Politician}: $c_L=l$\\
		& & $x^*_o=0$& \textit{Policy}: $x_L=\max\{k_l,V_r\}$\\
		& & $m^*_o=0$& \textit{Rent}: $m_L=0$\\
		\hline
		
		\multirow{4}{*}{\quad\textit{3}}& \multirow{4}{*}{$\quad\enspace\qquad\boldsymbol{V_r<0=V_o}$}& $x^*_l=0$& \textit{Party}: $P=L/R$\\
		& & $x^*_r=0$& \textit{Politician}: $c_P=l/o$\\
		& & $x^*_o=0$& \textit{Policy}: $x_P=0$\\
		& & $m^*_o=2\abs{\max\{k_l,V_r\}}$& \textit{Rent}: $m_P=0/2\abs{\max\{k_l,V_r\}}$\\
		\hline
		
		\multirow{4}{*}{\quad\textit{4}}& \multirow{4}{*}{$\quad\enspace\qquad\boldsymbol{V_o\leq0=V_r}$}& $x^*_l=0$& \textit{Party}: $P=L/R$\\
		& & $x^*_r=0$& \textit{Politician}: $c_P=l/r$\\
		& & $x^*_o=0$& \textit{Policy}: $x_P=0$\\
		& & $m^*_o=0$& \textit{Rent}: $m_P=0$\\
		\hline
		
		\multirow{5}{*}{\quad\textit{5}}& \multirow{2.5}{*}{$\quad\enspace\qquad\boldsymbol{V_r\leq0<V_o}$
		}& $x^*_l=0$& \textit{Party}: $P=R$\\
		& & $x^*_r=0$& \textit{Politician}: $c_R=o$\\
		&\qquad\qquad\enspace\quad\textit{and}& $x^*_o=\min\{k_o,V_o\}$& \textit{Policy}: $x_R=\min\{k_o,V_o\}$\\
		& \multirow{1.5}{*}{$\boldsymbol{\frac{\max\{k_l,V_r\}+\min\{k_o,V_o\}}{2}\leq b_R}$}& $m^*_o=\abs{\max\{k_l,V_r\}-b_R}$& \textit{Rent}: $m_R=\abs{\max\{k_l,V_r\}-b_R}$\\
		& & $\qquad\enspace-\abs{\min\{k_o,V_o\}-b_R}$&$\qquad\qquad\quad\enspace-\abs{\min\{k_o,V_o\}-b_R}$\\
		\hline
		
		\multirow{4}{*}{\quad\textit{6}}& \multirow{2}{*}{$\quad\enspace\qquad\boldsymbol{V_r\leq0<V_o}$}& $x^*_l=0$& \textit{Party}: $P=R$\\
		&\multirow{1.9}{*}{\qquad\qquad\enspace\quad\textit{and}}& $x^*_r=0$& \textit{Politician}: $c_R=o$\\
		&\multirow{2}{*}{$\boldsymbol{b_R<\frac{\max\{k_l,V_r\}+\min\{k_o,V_o\}}{2}}$}& $x^*_o=2b_R-\max\{k_l,V_r\}$& \textit{Policy}: $x_R=2b_R-\max\{k_l,V_r\}$\\
		& & $m^*_o=0$& \textit{Rent}: $m_R=0$\\
		\hline
		
		\multirow{4}{*}{\quad\textit{7}}& \multirow{2}{*}{$\quad\enspace\qquad\boldsymbol{V_o\leq0<V_r}$
		}& $x^*_l=0$& \textit{Party}: $P=R$\\
		& \multirow{1.9}{*}{\qquad\qquad\enspace\quad\textit{and}}& $x^*_r=\min\{k_r,V_r\}$& \textit{Politician}: $c_R=r$\\
		& \multirow{2}{*}{$\boldsymbol{\frac{\max\{k_l,V_o\}+\min\{k_r,V_r\}}{2}\leq b_R}$}& $x^*_o=0$& \textit{Policy}: $x_R=\min\{k_r,V_r\}$\\
		& & $m^*_o=0$& \textit{Rent}: $m_R=0$\\
		\hline

		\multirow{4}{*}{\quad\textit{8}}& \multirow{2}{*}{$\quad\enspace\qquad\boldsymbol{V_o\leq0<V_r}$
		}& $x^*_l=0$& \textit{Party}: $P=R$\\
		& \multirow{1.9}{*}{\qquad\qquad\enspace\quad\textit{and}}& $x^*_r=2b_R-\max\{k_l,V_o\}$& \textit{Politician}: $c_R=r$\\
		& \multirow{2}{*}{$\boldsymbol{b_R<\frac{\max\{k_l,V_o\}+\min\{k_r,V_r\}}{2}}$}& $x^*_o=0$& \textit{Policy}: $x_R=2b_R-\max\{k_l,V_o\}$\\
		& & $m^*_o=0$& \textit{Rent}: $m_R=0$\\
		\hline
		
		\multirow{5}{*}{\quad\textit{9}}& \multirow{2}{*}{$\enspace\quad\boldsymbol{0<V_r<b_R\leq V_o}\enspace\textit{or}$
		}& $x^*_l=0$& \textit{Party}: $P=R$\\
		& \multirow{2}{*}{$\quad\quad\boldsymbol{0<V_r<V_o<b_R,}$}& $x^*_r=0$& \textit{Politician}: $c_R=o$\\
		& \qquad\qquad\quad\enspace\multirow{2}{*}{\textit{and}}& $x^*_o=\min\{b_R,V_o\}$& \textit{Policy}: $x_R=\min\{b_R,V_o\}$\\
		& \multirow{2}{*}{$\enspace\qquad\qquad\boldsymbol{\bar{x}<b_R}$}& $m^*_o=\abs{\min\{b_R,V_r\}-b_R}$& \textit{Rent}: $m_R=\abs{\min\{b_R,V_r\}-b_R}$\\
		& & $\quad\quad\enspace-\abs{\min\{b_R,V_o\}-b_R}$ & $\qquad\qquad\quad\enspace-\abs{\min\{b_R,V_o\}-b_R}$\\
		\hline
		
		\multirow{4}{*}{\enspace\textit{10}}& \multirow{2}{*}{$\quad\enspace\quad\boldsymbol{0<V_o<V_r,b_R}$
		}& $x^*_l=0$& \textit{Party}: $P=R$\\
		& \multirow{1.9}{*}{\qquad\qquad\enspace\quad\textit{and}}& $x^*_r=\min\{b_R,V_r\}$& \textit{Politician}: $c_R=r$\\
		& \multirow{2}{*}{$\qquad\enspace\qquad\boldsymbol{\bar{x}< b_R}$}& $x^*_o=0$& \textit{Policy}: $x_R=\min\{b_R,V_r\}$\\
		& & $m^*_o=0$& \textit{Rent}: $m_R=0$\\

		\pagebreak
		
		\multirow{4}{*}{\enspace\textit{11}}& \multirow{4}{*}{$\quad\enspace\qquad\boldsymbol{b_R\leq V_o,V_r}$
		}& $x^*_l=0$& \textit{Party}: $P=R$\\
		& & $x^*_r=b_R$& \textit{Politician}: $c_R=r$\\
		& & $x^*_o=0$& \textit{Policy}: $x_R=b_R$\\
		& & $m^*_o=0$& \textit{Rent}: $m_R=0$\\
		\hline
	\end{longtable}
	\end{center}
where $\bar{x}=\frac{\min\{b_R,V_r\}+\min\{b_R,V_o\}}{2}$, and $m^*_l=m^*_r=0$ for all cases.
\end{theorem} 

Note that while some cases appear to be missing, they are not included in the table because they have no solution. For instance, the complement of Case 10 is $0<V_o<V_r,b_R$ and $b_R<\frac{\min\{b_R,V_r\}+\min\{b_R,V_o\}}{2}$. In this case, the intuitive optimal policies of $r$ and $o$ are those closest to $b_R$ within their grasp, since they stand the highest chance to please $R$,\footnote{The preference of the electorate is already given for $0<V_o,V_r$, so politicians need worry only about nomination.} so $x_r=\min\{b_R,V_r\}$ and $x_o=\min\{b_R,V_o\}$. However, since $x_o,x_r< b_R$, it is impossible for condition $2b_R<\min\{b_R,V_r\}+\min\{b_R,V_o\}$ to hold. The same reasoning applies to the absent cases.

\begin{proof}\ \\
We refer the reader to Appendix \ref{AppB}. 
\end{proof}

Roughly speaking, in Cases 1 and 2, neither $r$ nor $o$ stands a chance of winning, so $R$ picks the lesser of two evils as damage control. In Cases 3 and 4, $R$'s best scenario is a tie, and it is accomplished through one politician only. If possible, the politician seeks a positive rent for it. In Cases 5 through 8, only one politician has the necessary valence advantage to win surely. While $R$ is aware of the possibility of victory with her, it simply does not accept every platform offer she makes. If the demands of the prospective winner -either in the form of policy, rent, or both- outweigh the benefits of winning, then $R$ might prefer to endorse the other politician and lose the election altogether. In general, we say that the seemingly losing candidate, or the \textit{option}, has a \textit{balancing effect} if her most attractive platform offer has a constraining effect on that of the most advantaged politician. If it does not, then the prospective winner is able to propose the policy closest to her bliss point. In the aforementioned cases, the prospective winner always prefers to win herself than to let the other politician lose the election to the opposition. This follows from the fact that $k_r>0$ and $k_o\geq0$, so a non-negative policy is preferred to a negative one.  Cases 9 to 11 are similar in spirit, except for the fact that both politicians stand a chance of winning the election, so instead of proposing a policy close to their own bliss point, they propose the closest to $b_R$. Depending what the most attractive offers of candidates are, and their distance with respect to $b_R$, the balancing effect may or may not be observed.

\section{Discussion}\label{discussion}

Let us first study the implications of introducing an outsider to the game in terms of the winning policy. 

\begin{corollary}\ \\
If $V_r<0$, then introducing an outsider can preserve or make the winning policy more centrist. If $V_r=0$, then the policy remains the same or becomes more extreme. If $V_r>0$, then the effect on the policy is mixed, and ranges from preserving it to making it more extreme. The result depends on which politician (insider or outsider) has the upper hand, and her extremism relative to that of the party.
\end{corollary}

\begin{proof}\ \\
The proof is straightforwardly obtained comparing the policy $x_P$ of the winning match $(P, c_P, (x_P,m_P))$ given in Theorems \ref{Thm1} and \ref{Thm2} for each valence condition case.
\end{proof}

The corollary shows that conditional on $R$'s insider losing the election, that is $V_r<0$, introducing an outsider may reduce the winning policy's polarization. Indeed, by bringing in a more attractive politician in terms of relative valence (if not sufficiently so for winning), $R$ forces the opposition to propose a more centrist policy to counteract the outsider's appeal. If on the other hand, the outsider is less attractive than its insider, then $R$ simply does not endorse her, and the outcome is the same as that of the insiders only game.\footnote{Just as in $V_o\leq V_r<0$, this is also true for the cases $V_o\leq 0=V_r$, and $V_o\leq0<V_r$ where $\frac{\max\{k_l,V_o\}+\min\{k_r,V_r\}}{2}\leq b_R$.} If the insiders of both parties are equally perceived by voters, so $V_r=0$, then introducing a more appealing politician may actually make the policy more extreme. This is because the electorate's favor allows the outsider to get away with a policy close to her ideal point, which may well be non-centrist. If the insider stands a chance of winning and the outsider does not, then the balancing effect makes the policy less radical. If both candidates can potentially be elected, $0<V_o,V_r$, then in general the competition to please $R$ leads to more extreme policies than in the insider only game. Exceptions occur when the insider is in a position to make the closest policy offer to $R$, that is when $b_R\leq V_r$, and the party has a less extreme bliss point than $r$ does, $b_R<k_r$.

While for the insiders only model we obtain a similar result to Serra (2018), namely that the valence advantaged candidate sways the winning policy close to her bliss point (as close as her advantage allows), by introducing an outsider and competition for nomination, we obtain novel results. Notably, we show that under certain conditions, the winning policy may in fact become more centrist, even in the presence of relative valence advantage. The key aspect for deriving the result is the competition induced by the outside option, whether that is the party's insider or the outsider herself. This potentially limits policy swaying or excessive rent-seeking on the advantaged candidate's behalf. 

We now consider the effect of introducing an outsider to the game on the welfare of voters defined as follows. Let $U_M(P^i,c^i_P,(x^i_P,m^i_P))$ denote the utility the median voter $M$ derives from the winning match in the insiders-only game (hence \textit{i}), and let $U_M(P^o,c^o_P,(x^o_P,m^o_P))$ be the utility she derives from the winning match in the game with the outsider (\textit{o}). We say that there is a negative effect on the voters' welfare if $U_M(P^i,c^i_P,(x^i_P,m^i_P))>U_M(P^o,c^o_P,(x^o_P,m^o_P))$, a null effect if $U_M(P^i,c^i_P,(x^i_P,m^i_P))=U_M(P^o,c^o_P,(x^o_P,m^o_P))$, and a positive effect otherwise. In the same manner we define the welfare of party $R$.

\begin{corollary}\label{cor3} \ \\
If $V_r<0$, then introducing an outsider preserves or improves the welfare of voters. If $V_r=0$, then in the majority of cases, voter welfare is maintained or improved, with only one condition bringing about a reduction. If $V_r>0$ the effect on welfare is mixed and it depends on the valence of the outsider and the extremism of the insider. All the cases but one have a null or positive effect on the welfare of the nominating party.
\end{corollary}

\begin{proof}\ \\
The proof is straightforward. The cases where the effect on voter/party welfare is negative are listed in Appendix \ref{AppC}.
\end{proof}

There are a few points worth stressing. First, introducing an outsider may have a positive effect on voter welfare, particularly if competition with the insider is such that the balancing effect limits both policy polarization and rent-seeking; this occurs under certain conditions in Cases 6 and 8. This result is slightly similar to that of van Weelden (2013), where in equilibrium, elected candidates implement a policy slightly close to their bliss point and also engage in less than full rent seeking. However, the incentives for doing so are very different from those of our work. In van Weelden's model, the behavior is induced through the electorate's double threat: not re-electing the politician and voting for the opposition. Evidently, we do not have a re-election dynamics that allows for comparison, and while it is true that all politicians in our model are policy aware, they do not consider what happens in the period that follows the term in office. Our result is instead derived from the party's threat that it will run with its insider if the outsider asks for too much. Note also that there are other scenarios in which voter welfare increases, which correspond to having a competitive if not successful candidate against the opposition, and to a more moderate or less advantaged contender. The second point to note is that in nearly all the cases, the party's welfare is either preserved or increased. Notably, this is true even if the outsider engages in full rent, which is the positive amount that equalizes the utilities of the two platform offers. Loosely speaking, an outsider represents at best a brand new possibility to win, which as it turns out is better than losing, even if there is a price to pay. 

A last general remark is due. The results in both corollaries are independent of $k_o$, which implies that they hold for all types of outsiders: centrists, moderates and radicals. The reasoning is that more than ideology, what limits a politician is her relative valence advantage; the lesser her perception, the less room for platform negotiation. Ultimately, this means that partisan affect (and not extremism) is the variable that has a more significant effect on the welfare or voters and parties. 

\section{Conclusion}\label{conclusion}

With this work we add to the incipient literature on political outsiders. Ours is a theoretical approach aiming to elucidate under what conditions having such political figures is desirable for voters and parties. To this end, we expand Fowler and Hall's (2014) definition of partisan incumbency advantage, and allow broadly for partisan advantage (or disadvantage), so voters have some preference for each party, not just the incumbent. Our main result is that partisan affect has a higher effect on voter and party welfare than the extremism of the outsider. We also show that under certain conditions, outsiders may reduce swaying of the winning policy away from the median, which to our knowledge of the literature is a novel result. We also corroborate Serra's (2018) result on policy polarization under popular discontent. Lastly, we find that in some cases, the nominating party may be willing to pay the highest credibility cost which is full rent, in order to win the election.

There is a distinctive concept on the literature of neophytes that has remained unaccounted for in our work. Notably, we made no specific definition of an outsider other than being an individual with no prior political experience. Nonetheless, other authors distinguish between different types of amateurs based on their previous experience, whether outside or inside of politics (Roberds and Roberts, 2002; Green and Krasno, 1988; Porter and Treul, 2020). Adding this dimension to our analysis may bring about a change in the utility function of voters. To be specific, the transferable skills from certain professional backgrounds may be desirable for negotiating and implementing policies. This is related to Adams \textit{et al.}'s (2011) deconstruction of candidate valence into two types: that valued by voters for campaigning and winning elections, and that valued for ruling. While we have accounted for the former, we did not address candidate competence or dedication. In addition to modifying the utilities in the model, this approach may enable us to endogenize the outsider's decision to engage or not in politics. Apart from these considerations, we consider our work to be a good starting point for discussing the relation between partisan affect and the entry of outsiders to the political arena, and its implications for voter welfare.

\noindent \textbf{Declaration of interests}\ \\
The author declares that she has no known competing financial interests or personal relationships that could have appeared to influence the work reported in this paper.

\begin{appendices}
	\section{}\label{AppA}
	\textbf{Lemma 1.}
		\textit{In the presence of relative valence advantage $V_{c_R}\geq0$, party $R$ accepts policies $x_R$ in $[0,V_{c_R}]$, and $L$ accepts any $x_L\leq0$. The match $\left(L,l,(x_L,0)\right)$ loses the election to $\left(R,c_R,(x_R,m_R)\right)$, with $m_R\geq0$. In the presence of relative valence disadvantage $V_{c_R}<0$, $L$ accepts policies $x_L$ in $[V_{c_R},0]$, and $R$ accepts any $x_R\geq0$. In this case, the match $(R,c_R,(x_R,m_R))$ with $m_R\geq0$, loses the election to $(L,l,(x_L,0))$.}
		
	\begin{proof}\ \\
	Suppose that $V_{c_R}\geq0$, by definition this implies $\alpha_R\nu_R\geq\alpha_L\nu_L$, which means that $R$'s candidate has a valence advantage over $L$'s.\footnote{Strictly speaking, $c_R$ enjoys a clear-cut advantage when $V_{c_R}>0$.} This in turn implies that the candidate $c_R$ has a wider interval for successful policy proposing than $c_L$. To see why, suppose that $c_L$ and $c_R$ were to set an equidistant policy about the origin, so $x_L=-d$ and $x_R=d$, with $d>0$. Then, voters would be faced with the following utilities from electing each match: 
	\begin{eqnarray*}
		U_M\left(L,c_L,(-d,m_L)\right)&=&-d+\alpha_L\nu_L\\
		U_M\left(R,c_R,(d,m_R)\right)&=&-d+\alpha_R\nu_R.
	\end{eqnarray*}
	
	Since $\alpha_R\nu_R\geq\alpha_L\nu_L$, we have $U_M\left(L,c_L,(-d,m_L)\right)\leq U_M\left(R,c_R,(d,m_R)\right)$. Hence, if both candidates were to propose symmetric policies, then the valence advantage of $c_R$ would make her win the election. This further implies that in order to increase its competitiveness, $c_L$ has to propose a more centrist position than $c_R$ does. This can be seen geometrically in Figure \ref{fig:equilibria}. 
	
	\begin{figure}[h!]
		\centering
		\caption{Winning policies}
		\includegraphics[width=8cm]{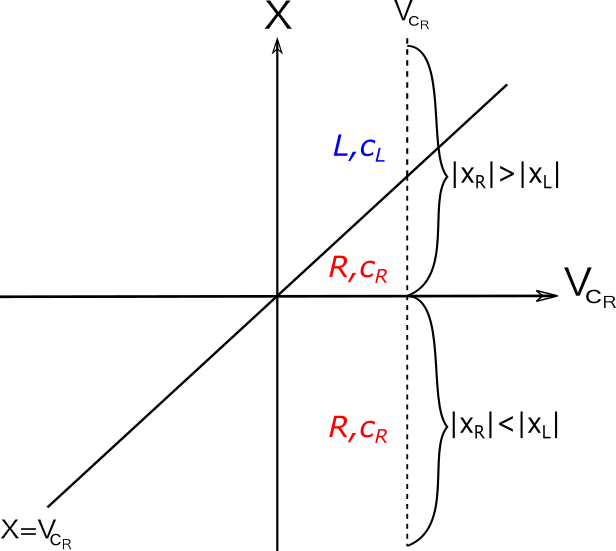}
		\label{fig:equilibria}
	\end{figure}
	
	Note that for a given value of $V_{c_R}\geq0$, if $\abs{x_R}<\abs{x_L}$, that is, if $c_L$ proposes a more radical policy than $c_R$ does, then $X<0$. This makes the overall utility of voters fall in the fourth quadrant of the $V_{c_R}X$ plane, which is where voters prefer the match $\left(R,c_R,(x_R,m_R)\right)$ over $\left(L,c_L,(x_L,m_L)\right)$. Likewise, if $\abs{x_L}=\abs{x_R}$, then overall utility falls on the $V_{c_R}$ axis, which is favorable to party $R$. It is only when $\abs{x_L}<\abs{x_R}$, that the region in which voters prefer the match $\left(L,c_L,(x_L,m_L)\right)$ can be reached.
	
	It is in $L$'s interest that the value of $X=\abs{x_R}-\abs{x_L}$ is high, as it brings the overall utility of voters to the region where its match is preferred. Indeed, the policy $x_L=0$ maximizes the value of $X$. Let us assume for now that this policy is in fact proposed to, and accepted by $L$, then $X=\abs{x_R}$. Note that any $0< x_R\leq V_{c_R}$ gives $R$ the victory,\footnote{Victory with certainty is achieved as long as $x_R<V_{c_R}$, so for the upper bound of the interval we take $x_R=V_{c_R}-\varepsilon$ with $\varepsilon>0$ infinitesimally small. For the assertion, we take the limit as $\varepsilon\rightarrow0$, and write $x_R$ in $[0,V_{c_R}]$.} and any $V_{c_R}<x_R$ grants $L$ the election. This ultimately implies that $R$ accepts a policy offer $x_R$ if $x_R$ in $[0,V_{c_R}]$.
	
	What is left to show is that $L$ accepts any $x_L\leq0$. This follows easily from the fact that $L$ is bound to lose the election even with its most centrist policy $x_L=0$. Since $L$ can do no better than this (recall that $x_L$ is required to be non-positive), it accepts any $x_L$ in $[0,\infty)$ knowing it will lose the election with certainty regardless. An analogous reasoning shows that if $V_{c_R}<0$, then $R$ accepts any $x_R\geq0$, and $L$ accepts $x_L$ in $[V_{c_R},0]$.
	\end{proof}

	\section{}\label{AppB}
	\textbf{Theorem 2.}
	\textit{A game consisting of insiders $l$ and $r$, and outsider $o$, has the equilibrium offers and winning match stated in Table \ref{tab:eq2}.}
	
	\begin{proof}\ \\
		
	\begin{enumerate}
		\item If $\boldsymbol{V_r\leq V_o<0}$, then by Lemma \ref{Lemma1}, $R$ is meant to lose the election, regardless of the politician it chooses to run with. Recall that in the proof of Theorem \ref{Thm1}, we adopted the convention that a politician bound to lose the nomination/election proposes a centrist policy; we extend this convention, and assume that a (prospective) losing politician proposes a null platform, so $(x_r,m_r)=(x_o,m_o)=(0,0)$. 
		
		The important point to note is that if $R$ were to run with $r$, then the winning policy of $L$ would be in $[V_r,0]$, and if it were to run with $o$ instead, then it would be in $[V_o,0]$, this again by Lemma \ref{Lemma1}. As a rule of thumb, $R$ picks the candidate that limits the range of successful policy proposing for the opposition, and since $[V_o,0]\subseteq[V_r,0]$, $R$ runs with $o$. Analysis similar to that in the proof of Theorem \ref{Thm1} shows that within $[V_o,0]$, $l$ proposes the policy closest to her bliss point, which is $x_l=\max\{k_l,V_o\}$. This is the result.
		
		\item The proof for $\boldsymbol{V_o\leq V_r<0}$ is analogous to that of Case 1.
		
		\item If $\boldsymbol{V_r<0=V_o}$ then by Lemma \ref{Lemma1}, and by an argument similar to that of Case 1, $R$ runs with $o$. By convention, $r$ proposes the platform $(x_r,m_r)=(0,0)$. Since $V_o=0$, it follows from Corollary \ref{cor1} that $x_o=0$. Knowing that the candidate $o$ runs with a centrist policy, $l$ has no option but to propose $0$ as well, as any other policy $x_l$ in $(-\infty,0)$ would lead to a sure loss. The rent sought by $l$ is, by definition, $m_l=0$.
		
		What is left is to show that $m_o=2\abs{\max\{k_l,V_r\}}$. First, note that $U_M(L,l,(0,0))=U_M(R,o,(0,m_o))=0$ for all $m_o\geq0$. Hence, by the indifference assumption of voters, the electorate chooses either match with equal probability. Therefore, the expected value of the outsider's platform offer $(x_o,m_o)=(0,m_o)$ to $R$ is $\mathbb{E}[U_R(0,m_o)]=\frac{1}{2}\left(-\abs{0-b_R}-m_o-\abs{0-b_R}\right)=-b_R-\frac{m_o}{2}$. The outsider chooses $m_o$ in such a way, that $R$ is (almost) indifferent between the aforementioned expected value, and the value of accepting $r$'s offer $(x_r,m_r)=(0,0)$, $U_R(0,0)=-\abs{\max\{k_l,V_r\}-b_R}=\max\{k_l,V_r\}-b_R$.\footnote{Recall that in this scenario, the fact that $V_r<0$ means that $(R,r,(0,0))$ loses to $(L,l,(\max\{k_l,V_r\},0))$ with certainty.} It is easy to check that $m_o=2\abs{\max\{k_l,V_r\}} -\varepsilon$ with $\varepsilon>0$ infinitesimally small, makes $R$ prefer the offer of $o$. For the result, we take the limit as $\varepsilon\rightarrow0$.
		
		\item The proof for $\boldsymbol{V_o\leq0=V_r}$ is analogous to that of Case 3, except $r$ has no leeway for proposing a positive rent. Note that while $R$ may receive two null platform offers, the indifference assumption of parties makes it accept the offer of $r$.
		
		\item If $\boldsymbol{V_r\leq0< V_o}$ and $\boldsymbol{\frac{1}{2}\left(\max\{k_l,V_r\}+\min\{k_o,V_o\}\right)\leq b_R}$, then just as in Theorem \ref{Thm1}, $r$ and $l$ propose $(x_l,m_l)=(x_r,m_r)=(0,0)$. $R$ accepts $o$'s platform proposal $(x_o,m_o)$ so long as two things occur. First, $x_o$ in $[0,V_o]$, this follows from Lemma \ref{Thm1}. Second, the rent sought by the outsider $m_o$ is such that $U_R(x_o,m_o)> U_R(x_r,0)$, otherwise the indifference assumption of parties would make $R$ nominate $r$ instead. Note that $U_R(x_r,0)=-|\max\{k_l,V_r\}-b_R|$, since endorsing $r$ leads to a sure loss.
		
		Within the interval $[0,V_o]$, the policy closest to to the outsider's bliss point is $x_o=\min\{k_o,V_o\}$. It is easy to check that $m_o=\abs{\max\{k_l,V_r\}-b_R}-\abs{\min\{k_o,V_o\}-b_R}-\varepsilon$ with $\varepsilon>0$ infinitesimally small, is such that $U_R(\min\{k_o,V_o\},\abs{\max\{k_l,V_r\}-b_R}-\abs{\min\{k_o,V_o\}-b_R})>U_R(x_r,0)$ with $x_r\geq0$. For the assertion, we take the limit as $\varepsilon\rightarrow0$ and write $m_o=\abs{\max\{k_l,V_r\}-b_R}-\abs{\min\{k_o,V_o\}-b_R}$.
		
		What is left to show is that $m_o$ does not violate the domain constraint $m_o\geq0$, this amounts to showing that $\abs{\max\{k_l,V_r\}-b_R}\geq\abs{\min\{k_o,V_o\}-b_R}$. It is a simple matter to check that the valence condition $\max\{k_l,V_r\}+\min\{k_o,V_o\}\leq 2b_R$, ensures just that. This is the result.
		
		\item If $\boldsymbol{V_r\leq0<V_o}$ and $\boldsymbol{b_R<\frac{1}{2}\left(\max\{k_l,V_r\}+\min\{k_o,V_o\}\right)}$, then $o$ is the only politician that ensures $R$ a victory, so as long as her platform offer $(x_o,m_o)$ is beneficial to $R$, at least more so than losing to the match $(L,l,(\max\{k_l,V_r\},0))$. We also have $(x_l,m_l)=(x_r,m_r)=(0,0)$. However, a reasoning similar to the above where $x_o=\min\{k_o,V_o\}$ and $m_o=\abs{\max\{k_l,V_r\}-b_R}-\abs{\min\{k_o,V_o\}-b_R}$, will not hold because the condition $2b_R<\max\{k_l,V_r\}+\min\{k_o,V_o\}$ yields a negative value of $m_o$. 
		
		Note that the condition $2b_R<\max\{k_l,V_r\}+\min\{k_o,V_o\}$ implies that the outsider's desired policy $\min\{k_o,V_o\}$ is too extreme for $R$ to accept. In fact, the most radical policy $R$ is willing to back up is $2b_R-\max\{k_l,V_r\}-\varepsilon$ with $\varepsilon>0$ infinitesimally small. This is indeed symmetric to $\max\{k_l,V_r\}$ about $b_R$, as Figure \ref{fig:utR} shows. Given that $\max\{k_l,V_r\}<2b_R-\max\{k_l,V_r\}<\min\{k_o,V_o\}$, $o$ derives a greater utility running with $R$ and implementing policy $x_o=2b_R-\max\{k_l,V_r\}$, than letting $r$ run instead, and losing the election to $l$, this is also shown in Figure \ref{fig:utR}. Note that this is true, even when the rent sought by the outsider is null. Thus, $o$ proposes $(2b_R-\max\{k_l,V_r\},0)$ wins nomination, and ultimately, the election. 
		
		
		\begin{figure}[h!]
			\centering
			\caption{Utilities of $o$ and $R$ when $b_R<\frac{1}{2}(\max\{k_l,V_r\}+\min\{k_o,V_o\})$}
			\includegraphics[width=11cm]{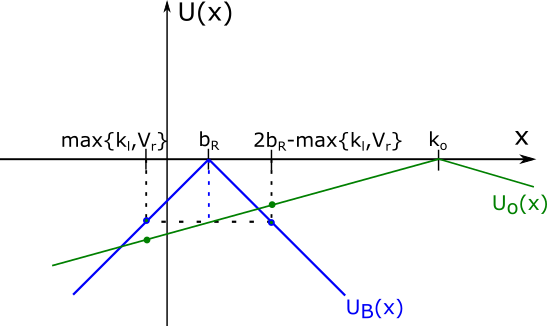}
			\label{fig:utR}
		\end{figure}
		
		\item The proof for $\boldsymbol{V_o\leq 0<V_r}$ and $\boldsymbol{\frac{1}{2}(\max\{k_l,V_o\}+\min\{k_r,V_r\})\leq b_R}$ is similar to that of 5.
		
		\item The proof for $\boldsymbol{V_o\leq 0<V_r}$ and $\boldsymbol{b_R<\frac{1}{2}(\max\{k_l,V_o\}+\min\{k_r,V_r\})}$ is similar to that of 6.
		
		\item If $\boldsymbol{0<V_r<b_R\leq V_o}$ or $\boldsymbol{0<V_r<V_o<b_R}$, and $\boldsymbol{\frac{1}{2}\left(\min\{b_R,V_r\}+\min\{b_R,V_o\}\right)<b_R}$, then both $r$ and $o$ stand a chance of winning the election. By Theorem \ref{Thm1} we have $(x_l,m_l)=(0,0)$. Given that the insider and outsider are in competition for nomination, they propose the policy closest to $R$'s bliss point within their respective ranges, the latter being derived from Lemma \ref{Lemma1}. Thus, $x_r=\min\{b_R,V_r\}$ and $x_o=\min\{b_R,V_o\}$. Note that since $V_r<b_R\leq V_o$ or $V_r<V_o<b_R$, the outsider's policy is always closest to that $R$ prefers. Knowing that $o$ has an advantage, $r$ proposes $(x_r,m_r)=(0,0)$.
		
		Lastly, the fact that $\min\{b_R,V_o\}<2b_R-\min\{b_R,V_r\}$ implies that there is a sufficient difference in the utility $R$ derives from the policies of $r$ and $o$,\footnote{Importantly, for the result we take $r$'s initial offer $x_r=\min\{b_R,V_r\}$ instead of the null.}  for the outsider to seek a positive rent. This is shown in Figure \ref{fig:utR9}. It is a simple matter to check that the rent $m_o=\abs{\min\{b_R,V_r\}-b_R}-\abs{\min\{k_o,V_o\}-b_R}-\varepsilon$ with $\varepsilon>0$ infinitesimally small is such that $R$ still prefers endorsing $o$ over $r$. Taking the limit as $\varepsilon\rightarrow0$ gives the result.
		
		
		\begin{figure}[h!]
			\centering
			\caption{Utility of $R$ when $0<V_r<b_R\leq V_o$ and $\min\{b_R,V_o\}<2b_R-V_r$}
			\includegraphics[width=8cm]{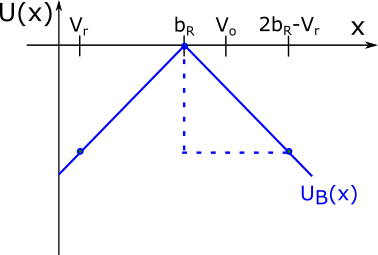}
			\label{fig:utR9}
		\end{figure}
		
		\item The proof for $\boldsymbol{0<V_o<V_r,b_R}$ and $\boldsymbol{\min\{b_R,V_r\}+\min\{b_R,V_o\}<2b_R}$ is a combination of arguments similar to those of Cases 5 and 9.
		
		\item If $\boldsymbol{b_R\leq V_o,V_r}$, then both politicians offer $R$ their most attractive (and attainable) policy, $x_r=x_o=b_R$. The indifference assumption of parties makes $R$ prefer the offer of its insider. Knowing this, the outsider offers the null platform $(x_o,m_o)=(0,0)$.
		
	\end{enumerate}
	\end{proof}

	\section{}\label{AppC}

	\textbf{Corollary 4.2.}
	\textit{If $V_r<0$, then introducing an outsider preserves or improves the welfare of voters. If $V_r=0$, then in the majority of cases, voter welfare is maintained or improved, with only one condition bringing about a reduction. If $V_r>0$ the effect on welfare is mixed and it depends on the valence of the outsider or the extremism of the insider. All the cases but one either maintain or improve the welfare of the party.}
	
	We provide a list of the cases where the introduction of an outsider to the game reduces welfare.
	\begin{itemize}
		\item Case 6. $V_r\leq 0<V_o$ and $b_R<\frac{\max\{k_l,V_r\}+\min\{k_o,V_o\}}{2}$. If $\nu_o$ is such that $\nu_o<\frac{\alpha_R\nu_r+\alpha_L\nu_l}{2}+2b_R$, then $U_M(R,r,(\min\{k_r,V_r\},0))>U_M(R,o,(2b_R,0))$.
		\item Case 9. $0<V_r<b_R\leq V_o$ and $\frac{\min\{b_R,V_r\}+\min\{b_R,V_o\}}{2}<b_R$. If $\nu_o$ and $k_r$ are such that $b_R\leq V_o<b_R-k_r+V_r$ and $k_r<V_r<b_R$ then $U_M(R,r,(\min\{k_r,V_r\},0))>U_M(R,o,(\min\{b_R,V_o\},\abs{\min\{b_R,V_r\}}-\abs{\min\{b_R,V_o\}-b_R}))$.
		\item Case 10. $0<V_o<V_r,b_R$ and $\frac{\min\{b_R,V_r\}+\min\{b_R,V_o\}}{2}<b_R$. If $\nu_o$ and $k_r$ are such that $k_r<V_o<V_r<b_R$, $V_o<k_r<V_r<b_R$, $k_r<V_o<b_R<V_r$, or $V_o<k_r<b_R<V_r$, then $U_M(R,r,(\min\{k_r,V_r\},0))>U_M(R,r,(\min\{b_R,V_r\},0))$.
		\item Case 11. $b_R\leq V_o,V_r$. If $k_r$ is such that $k_r<b_R\leq V_r$, then $U_M(R,r,(\min\{k_r,V_r\},0))>U_M(R,r,(b_R,0))$.
		On the other hand, if $\min\{k_r,V_r\}\neq b_R$, then $U_R(R,r,(\min\{k_r,V_r\},0))>U_R(R,r,(b_R,0))$.
	\end{itemize}
	
\end{appendices}

\end{document}